\documentclass[runningheads]{llncs}
\newif\ifFULL 
\FULLtrue
\usepackage[T1]{fontenc}
\usepackage{graphicx}
\usepackage{amssymb,amsmath,amsthm}
\usepackage{cleveref}
\usepackage{enumerate}
\crefformat{enumi}{(#2#1#3)}
\usepackage{mathrsfs}
\usepackage{multirow}
\usepackage{xspace}
\usepackage{paralist}

\usepackage{verbatim} 
\ifFULL
\newenvironment{onlyfull}{}{}\newenvironment{onlyproc}{\comment}{\endcomment}
\else
\newenvironment{onlyfull}{\comment}{\endcomment}\newenvironment{onlyproc}{}{}
\fi

\DeclareMathAlphabet\mathbfcal{OMS}{cmsy}{b}{n}

\theoremstyle{definition}

\newtheorem{theo}{Theorem}
\newtheorem{lemm}[theo]{Lemma}

\newtheorem{defi}[theo]{Definition}

\newcommand\labelx[1]{
  \label{#1}
  \newtheorem*{theo-#1}{\Cref{#1}}
}

\crefname{section}{Section}{Sections}
\crefname{definition}{Definition}{Definitions}
\crefname{lemma}{Lemma}{Lemmas}
\crefname{figure}{Fig.}{Figs.}
\crefname{table}{Table}{Tables}
\crefname{appendix}{Appendix}{Appendices}

\newcommand\ab\allowbreak
\newcommand\mc[3]{\multicolumn{#1}{#2}{#3}}

\newcommand\tsub[1]{\mathsf{Sub}(#1)}

\newcommand\ldepth{\mathsf{ld}}
\newcommand\minld{\mathsf{mld}}

\newcommand\redex{\mathsf{rdx}}

\newcommand\tuple[1]{\langle #1\rangle}
\newcommand\tto{\rightsquigarrow}
\newcommand\toi{\to^{\textsc{i}}}

\newcommand\Term[1]{\mathcal{T}_{#1}}
\newcommand\Ctx[1]{\mathcal{C}_{#1}}
\newcommand\Dom{\mathsf{dom}}
\newcommand\Lang{\mathcal{L}}
\newcommand\Var{\mathcal{V}}
\newcommand\FVar{\mathsf{FV}}
\newcommand\NF{\mathsf{NF}}
\newcommand\hole{\square}
\newcommand\A{{\scriptstyle @}}
\newcommand\Aa{{\scriptscriptstyle @}}
\newcommand\QF{Q\setminus\{q_F\}}

\newcommand\Qrdx{\mathscr{Q}}

\newcommand\Iff{\Leftrightarrow}
\newcommand\To{\Rightarrow}

\newcommand\cP{\mathbf{P}}
\newcommand\cPh{\mathbf{\Phi}}
\newcommand\cS{\mathbf{S}}
\newcommand\cD{{\mathbfcal{D}}}
\newcommand\cO{\mathbf{O}}

\newcommand\cK{\mathbf{K}}

\newcommand\cJ{\mathbf{J}}
\newcommand\cB{\mathbf{B}}
\newcommand\cC{\mathbf{C}}
\newcommand\cW{\mathbf{W}}

\newcommand\Nat{\mathbb{N}}

\newcommand\TDAlong{termination-disproving automaton\xspace}
\newcommand\TDA{TDA\xspace}

\newcommand\TDAs{TDAs\xspace}
\newcommand\TDASlong{termination-disproving automaton with a final sink\xspace}
\newcommand\TDAS{TDA-S\xspace}

\newcommand\TDASs{TDA-Ss\xspace}

\allowdisplaybreaks[3]

\usepackage[nomath]{kpfonts}

\begin{document}
\begin{onlyproc}
\title{Disproving Termination of Non-Erasing Sole Combinatory Calculus
with Tree Automata}
\end{onlyproc}
\begin{onlyfull}
\title{Disproving Termination of Non-Erasing Sole Combinatory Calculus
with Tree Automata \\
(Full Version)%
\thanks{This is the full version of the corresponding CIAA 2024 paper.}
}
\end{onlyfull}
%
\titlerunning{Disproving Termination of Non-Erasing Sole Combinatory Calculus with TA}
%
\author{Keisuke Nakano\inst{1}\orcidID{0000-0003-1955-4225} \\
Munehiro Iwami\inst{2}\orcidID{0000-0001-9925-450X}}
\authorrunning{K. Nakano and M. Iwami}
%
\institute{Tohoku University, Sendai, Miyagi, Japan \\
\email{k.nakano@acm.org} \\
\and
Shimane University, Matsue, Shimane, Japan\\
\email{munehiro@cis.shimane-u.ac.jp}
}
\maketitle              
\begin{abstract}
We study the termination of sole combinatory calculus,
which consists of only one combinator.
Specifically,
the termination for non-erasing combinators is disproven
by finding a desirable tree automaton with a SAT solver 
as done for term rewriting systems by Endrullis and Zantema.
We improved their technique to apply to non-erasing sole combinatory calculus,
in which it suffices to search for tree automata with a final sink state.
Our method succeeds in disproving the termination of 8 combinators,
whose termination has been an open problem.
\end{abstract}
%
%
%


\section{Introduction}

Combinatory logic~\cite{Schoenfinkel1924,Curry1930}
has been used in computer science
as a theoretical model of computation
and also as a basis for the design of functional programming languages%
~\cite{DBLP:journals/spe/Turner79,DBLP:books/ph/Jones87}.
It can be viewed as a variant of lambda calculus,
in which
a limited set of combinators,
primitive functions without free variables,
is used instead of lambda abstractions.

Combinators in combinatory logic are defined 
as \(Z x_1 x_2 \dots x_n \to e\)
where \(Z\) is a combinator,
\(x_1, \dots, x_n\) are variables,
and \(e\) is built from the variables with a function application.
It is known that a small set of combinators 
can define a combinatorial calculus that is sufficient
to cover all computable functions.
Well-known sets of such combinators are 
\(\{\cS, \cK\}\) and \(\{\cB, \cC, \cK, \cW\}\)
with 
\(\cS xyz \to xz(yz)\),
\(\cK xy \to x\),
\(\cB xyz \to x(yz)\),
\(\cC xyz \to xzy\), and 
\(\cW xy \to xyy\)~\cite{DBLP:books/daglib/0067558}.

The subject of this paper is \emph{sole combinatory calculus},
which consists of only one combinator.
There have been several studies on sole combinatory calculi.
Waldmann~\cite{DBLP:journals/iandc/Waldmann00} investigated the \(\cS\) combinator
to provide a procedure that decides whether an \(\cS\)-term,
built from \(\cS\) alone, has a normal form
and further showed that the set of normalizing \(\cS\)-terms is recognizable.
Probst and Studer~\cite{DBLP:journals/tcs/ProbstS01} proved that
the sole combinatory calculus with the \(\cJ\) combinator,
defined by \(\cJ xyzw \to xy(xwz)\)
is strongly normalizing; that is,
no \(\cJ\)-term has an infinite reduction sequence.
Ikebuchi and Nakano~\cite{DBLP:journals/lmcs/IkebuchiN20} 
showed that
the sole combinatory calculus with the \(\cB\) combinator
is strongly normalizing and characterized by equational axiomatization
for proving the looping and non-looping properties of repetitive right applications.

This paper concerns the non-termination of sole combinatory calculi,
where termination means that no term has an infinite reduction.
Let us say that a combinator is \emph{terminating}
when the corresponding sole combinatory calculus is terminating.
Iwami~\cite{weko_227536_1} has investigated the termination of 37 combinators
introduced in Smulyan's book~\cite{smullyan1994diagonalization}
and has reported that 10 of them shown in \Cref{fig:combs}
are of unknown termination.
\begin{figure}[t]
\begin{gather*}
\cP x y z \to z (x y z)
\quad\!
\cP_3 x y z \to y (x z y)
\quad\!
\cD_1 x y z w \to x z (y w) (x z)
\quad\!
\cD_2 x y z w \to x w (y z) (x w)
\\
\cPh x y z w \to x (y w) (z w)
\quad
\cPh_2  x y z  w_1 w_2 \to x  (y w_1 w_2)  (z w_1 w_2)
\quad
\cS_1 x y z w \to x y w (z w)
\\
\cS_2 x y z w \to x z w (y z w)
\qquad
\cS_3 x y z w v \to x y (z v) (w v)
\qquad
\cS_4 x y z w v \to z (x w v) (y w v)
\\[-6ex]
\end{gather*}
\caption{Combinators and their reduction rules}
\label{fig:combs}
~\\[-6ex]
\end{figure}

We disprove the termination of 8 of the 10 combinators.
%
The main idea of disproving the termination is 
to give a non-empty recognizable set of terms closed under the reduction
as Endrullis and Zantema have done~\cite{DBLP:conf/rta/EndrullisZ15}.
They showed that 
a SAT solver can find the corresponding tree automaton.
We improve their method
by showing that it suffices to search for tree automata with a final sink state
in our setting
and by reducing the number of variables in the SAT problem.
Our implemetation disproves the termination of 8 combinators.

\section{Preliminaries}
\label{sec:prelim}

A \emph{signature} (or \emph{alphabet}) \(\Sigma\)
is a non-empty finite set of \emph{function symbols},
each with a fixed natural number called \emph{arity}
(or \emph{rank})%
\footnote{
  We use the terminology of term rewriting systems,
  whereas definitions are given
  alongside that of formal language theory for the readers,
  e.g., a \emph{tree} and a \emph{rank} in formal language are called
  a \emph{(ground) term} and an \emph{arity} in term rewriting, respectively.
}.
The set of all function symbols of arity \(n\) in \(\Sigma\) is
written as \(\Sigma^{(n)}\).
We may write \(f^{(n)}\) 
for \(f\in\Sigma^{(n)}\).
A function symbol of arity 0 is called a \emph{constant} symbol.
A set of \emph{variables} is a countably infinite set
disjoint from \(\Sigma\).
For a set \(\Var\) of variables,
a set of \emph{terms} over \(\Sigma\), denoted by \(\Term{\Sigma}(\Var)\),
is inductively defined as the smallest set \(S\) 
such that \(\Var\subseteq S\) and 
\(t_1,\dots,t_n\in S\) implies \(f(t_1,\dots,t_n)\in S\)
for every \(f\in\Sigma^{(n)}\).
The set of variables occurring in \(t\in\Term{\Sigma}(\Var)\) is
denoted by \(\FVar(t)\).
In the rest of the paper,
the set \(\Var\) of variables is fixed
and contains \(x, x_1, x_2, \dots\) as its elements.
A \emph{substitution} is a finite map from variables to terms.
We write \(\Dom(\alpha)\) for the domain of a substitution \(\alpha\).
For a term \(t\) and a substitution \(\alpha\),
we denote by \(t\alpha\) an \emph{instance} of \(t\), 
a term obtained by replacing every variable \(x\) in \(t\) with \(\alpha(x)\).
Substitutions may be represented in the set notation as usual:
we write 
\(\{x_1\mapsto t_1,\dots, x_n\mapsto t_n\}\) 
for substitution \(\alpha\) 
when \(\Dom(\alpha)=\{x_1,\dots,x_n\}\) and
\(\alpha(x_i)=t_i\) holds for each \(i\)
and, in particular, \(\emptyset\) for substitution \(\alpha\) 
when \(\Dom(\alpha)=\emptyset\).
A term containing no variables is called a \emph{ground term},
and the set of ground terms is written as \(\Term{\Sigma}\),
i.e., \(\Term{\Sigma}=\Term{\Sigma}(\emptyset)\).
%
A
\emph{context} \(C\) is
a term over \(\Sigma\cup\{\hole\}\)
with a constant symbol \(\hole\), called a \emph{hole},
such that \(\hole\) occurs exactly once in \(C\).
For a context \(C\) and a term \(t\),
we denote by \(C[t]\) a term obtained by replacing the hole in \(C\) with \(t\).
We write \(\Ctx{\Sigma}(\Var)\) for the set of contexts;
in particular, 
\(\Var\) is omitted from the notation if it is empty,
i.e., \(\Ctx{\Sigma}=\Ctx{\Sigma}(\emptyset)\).
A term \(s\in\Term{\Sigma}\) is a \emph{subterm} of \(t\in\Term{\Sigma}\)
if \(t=C[s]\) holds with some \(C\in\Ctx{\Sigma}\).
%
The set of all subterms of \(t\) is written
as \(\tsub{t}\).

A \emph{term rewriting system} (TRS) \(R\) over \(\Sigma\) is
a set of rewriting rules of the form \(l\to r\) with 
\(l, r\in\Term{\Sigma}(\Var)\) where \(\FVar(r)\subseteq\FVar(l)\).
A TRS \(R\) is said to be \emph{non-erasing}
if \(\FVar(r)=\FVar(l)\) holds for every rule \(l\to r\in R\).
A TRS \(R\) is said to be \emph{left-linear}
if each variable in \(\FVar(l)\) occurs exactly once in \(l\)
for every rule \(l\to r\in R\).
A left-linear TRS \(R\) is said to be \emph{orthogonal}
if every pair of two (possibly the same) rules
in \(R\) has no overlapping,
which means that
the left-hand side of one rule is not unifiable with
any non-variable subterm of the left-hand side of the other rule.
%
A \emph{rewrite relation} \(\to_R\) over \(\Term{\Sigma}(\Var)\) induced by \(R\)
is defined by \(\{(C[l\alpha], C[r\alpha]) 
\mid C\in\Ctx{\Sigma}(\Var),~ l\to r\in R,~
\alpha: \FVar(l) \to\Term{\Sigma}(\Var)\}\).
The subscript \(R\) may often be omitted if no confusion arises.
A term \(t\) is called a \emph{redex} of \(R\)
if \(t=l\alpha\) holds for some \(l\to r\in R\) and substitution \(\alpha\);
that is, a redex means a reducible part of a term.
A term \(t\) is said to be in \emph{normal form} with respect to \(R\)
if there is no term \(u\) such that \(t\to_R u\);
in other words, no subterm of \(t\) is a redex of \(R\).
A set of normal forms with respect to \(R\) is denoted by \(\NF(R)\).
A TRS \(R\) is \emph{terminating} or \emph{strongly normalizing}
if no infinite rewrite sequence \(t_0\to_R t_1\to_R t_2\to_R \dots\)
with \(t_i\in\Term{\Sigma}(\Var)\) and \(i\in\Nat\)
exists.
A TRS \(R\) is \emph{weakly normalizing}
if for every term \(t\) there exists a term \(u\in\NF(R)\)
such that \(t\to_R^* u\) holds.
A rewrite step \(s \to_R t\) is \emph{innermost},
denoted by \(s \toi_R t\),
if no proper subterm of the contracted redex is itself a redex%
~\cite{DBLP:books/daglib/0007479},
that is,
the relation \(\toi_R\) is
defined by a subset of \(\to_R\) 
as \(\{(C[l\alpha], C[r\alpha]) 
\mid C\in\Ctx{\Sigma}(\Var),~ l\to r\in R,~
\alpha: \FVar(l) \to\Term{\Sigma}(\Var),~
\tsub{l\alpha}\subseteq\NF(R)\cup\{l\alpha\} \}\).
A TRS \(R\) is \emph{weakly innermost normalizing}
if for every term \(t\) there exists a term \(u\in\NF(R)\)
such that \(t\mathbin{\toi_R}^* u\) holds.
A set \(S\) of terms is \emph{closed} under \(R\)
if for every \(t\in S\), \(t\to_R u\) implies \(u\in S\).

A (non-deterministic bottom-up) \emph{tree automaton} 
is a quadruple \(A = \tuple{Q, \Sigma,\ab F, \Delta}\)
where \(Q\) is a finite set of states,
\(F\subseteq Q\) is a set of final states, and
\(\Delta\) is a set of state transition rules of the form \(f(q_1, \dots, q_n)\tto q\)
with \(f\in\Sigma^{(n)}\), \(q_1,\dots,q_n,q\in Q\).
The arrow \(\tto_\Delta\) is used for the rewrite relation over \(\Term{\Sigma\cup Q}\)
induced by the rules in \(\Delta\)
where the subscript \(\Delta\) may often be omitted if no confusion arises.
The set \(\Lang(A, q)\) for \(q\in Q\) is defined by
\(\Lang(A, q)\equiv\{ t \in \Term{\Sigma} \mid t \tto_\Delta^* q \}\).
The set of terms \emph{accepted} by \(A\)
is defined by 
\(\Lang(A)\equiv\bigcup_{q\in F} \Lang(A, q)\).
A state \(q\) is said to be \emph{reachable}
if \(\Lang(A, q)\) is not empty.
A state \(q\) is called a \emph{sink state}
when,
for every \(f\in\Sigma^{(n)}\) (\(n>0\)) and \(q', q_1,\dots,q_n\in Q\)
with \(q_i = q\) for some \(i\),
\(f(q_1,\dots,q_n)\tto q\in\Delta\) holds but
\(f(q_1,\dots,q_n)\tto q'\in\Delta\) does not hold
with \(q'\ne q\).
When \(q\) is a sink state,
it is easy to show that \(t\in \Lang(A, q)\) implies \(C[t]\in\Lang(A, q)\)
for every context \(C\in\Ctx{\Sigma}\) and every ground term \(t\in\Term{\Sigma}\).
Two tree automata \(A_1\) and \(A_2\) are said to be \emph{equivalent}
if \(\Lang(A_1)=\Lang(A_2)\) holds.

\section{Termination of sole combinatory calculus}
A combinatory calculus is specified
by certain kinds of combinators and reduction rules for each combinator.
A reduction rule for a combinator \(Z\) has the form
\( Z ~x_1 ~\dots ~x_n \to e \)
where \(e\) is built by combining \(x_1, \dots, x_n\) with function application.
One of the most familiar combinatory calculi is given by
the \(\cS\) and \(\cK\) combinators
defined by
\(\cS ~x_1 ~x_2 ~x_3 \to x_1 ~x_3 ~(x_2 ~x_3)\) and 
\(\cK ~x_1 ~x_2 \to x_1\).
The calculus is well known to be Turing-complete
in the sense that
these two combinators are sufficient
to represent all computable functions.
Reductions in combinatory calculus are easily simulated by a TRS
using the set of constant symbols for combinators 
and a binary function symbol \(\A\) for function application.
For example, 
reduction rules for the \(\cS\) and \(\cK\) combinatory calculus
can be represented by a TRS consisting of 
\(\A(\A(\A(\cS, x_1), x_2), x_3) \to \A(\A(x_1, x_3), \A(x_2, x_3))\) and
\(\A(\A(\cK, x_1), x_2) \to x_1\),
which is obviously non-terminating
because of the Turing-completeness of the corresponding reduction system.

In this paper, we are interested in the question for sole combinatory calculus,
which is built only by one combinator.
We start with formal definitions of several notions 
on the sole combinatory calculus in terms of term rewriting.
For \(Z\) is a combinator,
we denote by \(\Sigma_Z\) a signature consisting of
a constant \(Z\) and binary function symbol \(\A\).
A \emph{\(Z\)-term}, which is built only from \(Z\) in combinatory calculus,
is represented by a term in \(\Term{\Sigma}\).
A \emph{sole combinatory calculus} \(R_Z\) induced by a combinator \(Z\)
is a TRS over \(\Term{\Sigma_Z}\)
where \(R_Z\) is a singleton set of a left-linear rule of the form
\( \A(\dots\A(\A(Z, x_1), x_2)\dots, x_n) \to e \)
with a term \(e\) built only from \(\A\) and variables \(x_1,\dots,x_n\).
A combinator \(Z\) is said to be \emph{terminating}
if \(R_Z\) is terminating. 
It is known that
\(\cB\)~\cite{DBLP:journals/lmcs/IkebuchiN20}
and \(\cJ\)~\cite{DBLP:journals/tcs/ProbstS01} are terminating while
\(\cO\)~\cite{Klop07,weko_60631_1,DBLP:conf/rta/EndrullisZ15} and
\(\cS\)~\cite{DBLP:journals/iandc/Waldmann00} are not.

The following lemma allows us to consider
only the case of weakly innermost normalizing instead of terminating
since the rule of sole combinatory calculus is orthogonal.
\begin{lemm}[{\cite[Theorem 11]{DBLP:books/sp/ODonnell77}}]\label{lem:win-sn}
An orthogonal TRS \(R\) is terminating
if and only if \(R\) is weakly innermost normalizing.
\end{lemm}

The readers might recall a similar result shown by Church~\cite{10.5555/1096495} that
an orthogonal non-erasing TRS \(R\) is terminating
if and only if \(R\) is weakly normalizing.
Since we are concerned with non-erasing combinatory calculus,
this result may seem more convenient.
However, in the context of the present work, we intend to employ the above lemma
because the `innermost' condition plays a crucial role in our method
which will be detailed later.

\section{Disproving termination of combinators by tree automata}
\label{sec:disproof}
Endrullis and Zantema have proposed a procedure 
for disproving weakly and strongly normalizing
by finding tree automata 
that disprove termination of arbitrary forms of left-linear TRSs.
This section first explains how to disprove termination 
using the search for tree automata,
and shows that it is sufficient to find a restricted form of tree automata
to disprove the termination of non-erasing combinators.
Then, we present how to find such tree automata with a SAT solver.
Finally, 
we show the non-termination of 8 combinators with our implementation of the method.

\subsection{Disproving termination with tree automata}
The idea of disproving termination of a TRS \(R\) by Endrullis and Zantema
is to find a non-empty set of reducible ground terms (i.e., not in normal form),
which is closed under the rules in \(R\).
It is easy to see that the existence of such a set implies 
that \(R\) is not weakly normalizing
because any reduction from a term in the set is always infinite.
Endrullis and Zantema considered the case 
where the set is recognizable by a tree automaton as defined below,
though they did not give a name to it.
\begin{defi}\label{def:tdta}
  A \emph{\TDAlong (\TDA)} for a left-linear TRS \(R\) is
  a tree automaton \(A=\tuple{Q, \Sigma, F, \Delta}\) such that
  \begin{inparaenum}[(\ref{def:tdta}-1)]
  \item\label{item:ta-non-empty}%
  there exists a state \(q\in F\) that is reachable,
  \item\label{item:ta-no-nf}%
  \(\Lang(A)\cap\NF(R)=\emptyset\), and
  \item\label{item:ta-R-closed}%
  \(t\in\Lang(A, q)\) implies \(u\in\Lang(A, q)\)
  for all \(q\in Q\) and \(t, u\in\Term{\Sigma}\) with \(t\to_R u\).
  \end{inparaenum}
\end{defi}
\newcommand\reftdta[1]{(\ref{def:tdta}-\ref{item:ta-#1})}

It is easy to see that the set \(\Lang(A)\) for a \TDA \(A\) for \(R\)
can disprove weak normalizability of \(R\):
the conditions \reftdta{non-empty} and \reftdta{no-nf}
allow us to choose a term not in normal form,
and the conditions \reftdta{no-nf} and \reftdta{R-closed}
force any reduction from the term to be infinite.
The condition \reftdta{R-closed} is a bit strong in the sense that
it requires every set \(\Lang(A,q)\) with \(q\in Q\)
to be closed under reductions in \(R\).
The last condition might be relaxed to require the closure property
only for final states in \(Q_F\).
However, we require it for all states in \(Q\)
to make the SAT solver--based disproof search easier
as done by Endrullis and Zantema.
%
%

As we will see later,
it is sufficient to find a \TDA with a sink state as the final state.
The restricted form of tree automata is defined as follows.
\begin{defi}\label{def:tdtas}
  A \emph{\TDASlong (\TDAS)} for an orthogonal TRS \(R\) is
  a tree automaton \(A=\tuple{Q, \Sigma, \{q_F\}, \Delta}\) 
  with \(q_F\in Q\) such that
  \begin{inparaenum}[(\ref{def:tdtas}-1)]
  \item\label{item:tas-non-empty}%
  \(q_F\) is sink and reachable,
  \item\label{item:tas-no-nf}%
  \(\Lang(A)\cap\NF(R)=\emptyset\), and
  \item\label{item:tas-R-closed}%
  \(t\in\Lang(A,q)\) implies \(u\in\Lang(A,q)\cup\Lang(A)\)
  for all \(q\in Q\) and \(t, u\in\Term{\Sigma}\) with \(t\toi_R u\).
  \end{inparaenum}
\end{defi}
\newcommand\reftdtas[1]{(\ref{def:tdtas}-\ref{item:tas-#1})}
We assume here that \(R\) is orthogonal, which is stronger than left-linear,
because every sole combinatory calculus is orthogonal.
A major difference from \TDAs is that
a \TDAS has a sink state \(q_F\) as the unique final state.
The sink state is reachable so that \(\Lang(A)\) is non-empty.
In addition,
it requires the closure property only under the innermost reduction
and allows the reduction onto \(\Lang(A)\).
The latter relaxation corresponds to a minor improvement 
done by Endrullis and Zantema
where the closure property allows the reduction onto
terms accepted by the `larger' states
in the total order over states,
which will be used for reachability checking with \Cref{lem:non-empty}.
The following theorem states that
the existence of a \TDAS for an orthogonal TRS \(R\)
implies that \(R\) is not terminating.
%
%
\begin{theo}\labelx{thr:tdta}
  An orthogonal TRS \(R\) is not terminating
  if a \TDAS for \(R\) exists.
\end{theo}
\begin{proof}
Let \(A=\tuple{Q,\Sigma,\{q_F\},\Delta}\) be a \TDAS 
for an orthogonal TRS \(R\).
By \Cref{lem:win-sn}, 
it suffices to show that \(R\) is not weakly innermost normalizing.
We first prove by contradiction that
\begin{quote}
  (*) for all \(t\in\Lang(A)\), 
  no reduction sequence from \(t\) yields a normal form,
\end{quote}
which disproves the weak innermost normalizability of \(R\)
due to the non-emptiness of \(\Lang(A)\)
implied by the~\reftdtas{non-empty} condition of \(A\).
Suppose that
there exists a term such that
a reduction sequence starting from the term yields a normal form.
Let \(t\) be such a term which has the shortest reduction sequence 
\(t=t_0\toi_R t_1\toi_R\cdots\toi_R t_n\) with a normal form \(t_n\).
From the~\reftdtas{no-nf} condition, 
\(t\) is not a normal form, i.e., \(n>0\).
From the~\reftdtas{R-closed} condition with \(q=q_F\),
\(t_1\in\Lang(A)\) holds.
This contradict the assumption that
the reduction sequence from \(t\) is the shortest one.
Therefore, (*) holds.
\end{proof}

We will try to find a \TDAS (which has exactly one final state)
to disprove the termination of a sole combinatory calculus
instead of a \TDA.
Since the disproof search will be done by fixing the number of states
and increasing it iteratively,
we have to show that
the number of states of a \TDAS is not required to be as large as that of a \TDA.
Note that we do not need to find a \TDAS equivalent to a \TDA if it exists.
It is well-known that
every non-deterministic tree automaton 
can be converted into an equivalent one that has exactly one final state
(by introducing a fresh final state)
but it may have one more state than the original one.
The following lemma guarantees that
the number of states of a \TDAS is not required to be larger than that of a \TDA
to disprove the termination of an orthogonal TRS.
The proof idea is 
to construct a \TDAS \(A\) from a given \TDA \(A_0\)
by forcing a final state of \(A_0\) to be the final sink state of \(A\)
and removing states that accept only terms
whose subterm is accepted by the sink state.
%
\begin{lemm}\labelx{lem:tdta-sink}
  Let \(R\) be an orthogonal TRS.
  If a \TDA \(A_0=\tuple{Q_0,\Sigma,F_0,\Delta_0}\) for \(R\) exists,
  then a \TDAS \(A=\tuple{Q,\Sigma,\{q_F\},\Delta}\) for \(R\) also exists
  with \(|Q|\le|Q_0|\).
\end{lemm}
\begin{proof}
Let \(A_0=\tuple{Q_0,\Sigma,F_0,\Delta_0}\) be a \TDA
for an orthogonal TRS \(R\),
where \(A_0\) satisfies 
the conditions \reftdta{non-empty}, \reftdta{no-nf}, and \reftdta{R-closed}.
Without loss of generality, all states in \(Q_0\) are reachable;
otherwise, we could remove unreachable states from \(Q_0\).
From \reftdta{non-empty} of \(A_0\), the set \(F_0\) is not empty,
hence we can choose a final state \(q_F\in F_0\).
Let \(L_F\subseteq\Term{\Sigma}\) and \(Q_F\subseteq Q_0\) be defined by
\(L_F = \{ C[t] \mid C\in\Ctx{\Sigma},~ t\in\Lang(A_0,q_F) \}\) and
\(Q_F = \{ q\in Q_0 \mid \Lang(A_0,q)\subseteq L_F\}\), respectively.
Note that \(q_F\in Q_F\) holds in particular.
Then we define a tree automaton
\(A=\tuple{Q,\Sigma,\{q_F\},\Delta}\) 
with \(Q=(Q_0\setminus Q_F)\cup\{q_F\}\) and \(\Delta=\Delta_1\cup\Delta_2\)
where
\begin{align*}
  \Delta_1 &=
  \{ f(q_1,\dots,q_n)\tto q \in\Delta_0 \mid 
     \{q_1,\dots,q_n\}\subseteq Q_0\setminus Q_F,~ q\in Q \}~
  \text{and}\\
  \Delta_2 &=
  \{ f(q_1,\dots,q_n)\tto q_F \mid f\in\Sigma^{(n)},~ 
     q_F\in\{q_1,\dots,q_n\}\subseteq Q \}
  \end{align*}
so that \(q_F\) is to be a sink final state in \(A\).
Since \(|Q|\le|Q_0|\) obviously holds, 
it suffices to show that \(A\) is a \TDAS for \(R\).

Before proving that \(A\) is a \TDAS for \(R\), we show 
\begin{enumerate}[~(\ref{lem:tdta-sink}-1)]
  \item\label{tdta-posLF}
  \(t\in L_F\) if and only if \(t\in\Lang(A, q_F)\),
  \item\label{tdta-midLF}
  \(t\in\Lang(A, q)\) implies \(t\in\Lang(A_0, q)\)
  for all \(q\in Q_0\setminus Q_F\), and
  \item\label{tdta-negLF}
  \(t\in\Lang(A_0, q)\) implies \(t\in\Lang(A, q)\cup\Lang(A,q_F)\)
  for all \(q\in Q_0\setminus Q_F\),
\end{enumerate}
\newcommand\refLF[1]{(\ref{lem:tdta-sink}-\ref{tdta-#1})}
for all \(t\in\Term{\Sigma}\).
These statements can be shown by simultaneous induction on the structure of \(t\).
Suppose that \(t=f(t_1,\dots,t_n)\) with
\(f\in\Sigma^{(n)}\) and \(t_1,\dots,t_n\in\Term{\Sigma}\).
On the~\refLF{posLF} statement,
we examine two cases according to whether \(t_i\in L_F\) or not for each \(i\).
In the case where \(t_i\in L_F\) holds for some \(1\le i\le n\),
the `if'-statement of \refLF{posLF} is obvious
from the definition of \(L_F\),
hence we show the `only if'-statement.
Since \(t_i\in\Lang(A,q_F)\) holds from the induction hypothesis, 
we have \(t\in\Lang(A,q_F)\) using a transition rule in \(\Delta_2\).
Thus, the `only if'-statement of \refLF{posLF} also holds.
In the case where \(t_i\not\in L_F\) holds for all \(1\le i\le n\),
we first show the `if'-statement of \refLF{posLF}.
Assume \(t\in\Lang(A,q_F)\) holds.
Then, there exists \(q_1,\dots,q_n\in Q\) such that
\(f(q_1,\dots,q_n)\tto q_F\in\Delta\) and
\(t_i\in\Lang(A,q_i)\) for every \(1\le i\le n\).
If \(q_i=q_F\) holds for some \(1\le i\le n\),
then we have \(t_i\in L_F\) from the induction hypothesis,
hence \(t\in L_F\).
If \(q_i\in Q_0\setminus Q_F\) holds for all \(1\le i\le n\),
then the transition rule is in \(\Delta_1\)
and we have \(t_i\in\Lang(A_0,q_i)\) from the induction hypothesis of \refLF{midLF}
for all \(i\).
Using the same transition rule, we have \(t\in\Lang(A_0,q_F)\), hence \(t\in L_F\).
Therefore, the `if'-statement of \refLF{posLF} holds.
For the `only if'-statement of \refLF{posLF},
assume \(t\in L_F\).
Since \(t_i\not\in L_F\) holds for all \(1\le i\le n\),
we have \(t\in\Lang(A_0,q_F)\) owing to the definition of \(L_F\).
Then, there exists \(q_1,\dots,q_n\in Q_0\) such that
\(f(q_1,\dots,q_n)\tto q_F\in\Delta_0\) and
\(t_i\in\Lang(A_0,q_i)\) for every \(1\le i\le n\).
Note that \(t_i\not\in L_F\) implies \(q_i\in Q_0\setminus Q_F\)
by the definition of \(Q_F\).
Thus, the transition rule is in \(\Delta_1\) and
we have \(t_i\in\Lang(A,q_i)\cup\Lang(A,q_F)\)
from the induction hypothesis of \refLF{negLF}.
When \(t_i\in\Lang(A,q_i)\) holds for all \(i\),
we have \(t\in\Lang(A,q_F)\) using the same transition rule.
When \(t_i\in\Lang(A,q_F)\) holds for some \(1\le i\le n\),
we have \(t\in\Lang(A,q_F)\) using the transition rule in \(\Delta_2\).
Therefore, the `only if'-statement of \refLF{posLF} holds.
%

On the~\refLF{midLF} statement,
assume \(t\in\Lang(A,q)\) with \(q\in Q_0\setminus Q_F\),
that is, \(q\ne q_F\).
Then, there exist \(q_1,\dots,q_n\in Q\) such that
\(f(q_1,\dots,q_n)\tto q\in\Delta\) and
\(t_i\in\Lang(A,q_i)\) for every \(1\le i\le n\).
Note that \(q_i\ne q_F\) holds for all \(1\le i\le n\) 
because of the construction of \(\Delta\).
Since we have \(q_i\in Q_0\setminus Q_F\) for every \(1\le i\le n\),
the transition rule is in \(\Delta_1\) and
\(t_i\in\Lang(A_0,q_i)\) holds from the induction hypothesis
for every \(1\le i\le n\).
Using the same transition rule, we have \(t\in\Lang(A_0,q)\).
Therefore, \refLF{midLF} holds.
%

On the~\refLF{negLF} statement,
assume \(t\in\Lang(A_0,q)\) with \(q\in Q_0\setminus Q_F\).
Then,
there exist \(q_1,\dots,q_n\in Q_0\) such that
\(f(q_1,\dots,q_n)\tto q\in\Delta_0\) and
\(t_i\in\Lang(A_0,q_i)\) for every \(1\le i\le n\).
When \(q_i\in Q_F\) holds for some \(1\le i\le n\),
we have \(t_i\in L_F\) by the definition of \(Q_F\),
hence \(t_i\in\Lang(A,q_F)\) holds from the induction hypothesis of \refLF{posLF}.
Using a transition rule in \(\Delta_2\), we have \(t\in\Lang(A,q_F)\).
When \(q_i\in Q_0\setminus Q_F\) holds for all \(1\le i\le n\),
the transition rule is in \(\Delta_0\) and
we have \(t_i\in\Lang(A,q_i)\cup\Lang(A,q_F)\) from the induction hypothesis
for all \(1\le i\le n\).
If \(t_i\in\Lang(A,q_i)\) holds for all \(i\),
then we have \(t\in\Lang(A,q)\) using the same transition rule.
If \(t_i\in\Lang(A,q_F)\) holds for some \(i\),
then we have \(t\in\Lang(A,q_F)\) using the transition rule in \(\Delta_2\).
Therefore, \refLF{negLF} holds.

Now we are ready to show that \(A\) is a \TDAS for \(R\).
Concerning the~\reftdtas{non-empty} condition,
the set \(L_F\) is not empty since \(q_F\) is reachable in \(A_0\)
due to the~\reftdta{non-empty} condition of \(A_0\).
Then we have \(q_F\) is reachable also in \(A\) owing to \refLF{posLF}.
In addition, \(q_F\) is a sink state in \(A\), 
hence the~\reftdtas{non-empty} condition holds for \(A\).

The~\reftdtas{no-nf} condition is shown by contradiction.
Suppose that there exists a term \(t\in\Lang(A)\cap\NF(R)\).
Then we have \(t\in L_F\) from \refLF{posLF},
hence \(t=C[t_0]\) holds
for some \(C\in\Ctx{\Sigma}\) and \(t_0\in\Lang(A_0,q_F)\).
From \(t_0\in\Lang(A_0,q_F)\subseteq\Lang(A_0)\) 
and the~\reftdta{no-nf} condition of \(A_0\),
we have \(t_0\not\in\NF(R)\),
and thus \(t=C[t_0]\not\in\NF(R)\) holds.
This contradicts the assumption,
hence the~\reftdtas{no-nf} condition holds for \(A\).

Concerning the~\reftdtas{R-closed} condition,
assume that we have 
\(t\in\Lang(A,q)\) with \(q\in Q\)
and \(t\toi_R u\) for some \(u\in\Term{\Sigma}\).
In the case of \(q=q_F\),
we have \(t\in L_F\) by \refLF{posLF},
hence \(t=C[t_0]\) holds for some \(C\in\Ctx{\Sigma}\) and \(t_0\in\Lang(A_0,q_F)\).
Since the \(t_0\) is not in normal form by the~\reftdta{no-nf} condition of \(A_0\)
and \(\toi_R\) is an innermost relation,
we have either \(u=C[u_0]\) with \(t_0\toi_R u_0\in\Term{\Sigma}\)
or \(u=D[t_0]\) with some \(D\in\Ctx{\Sigma}\).
In the former case,
we have \(u_0\in\Lang(A_0,q_F)\) 
by \(t_0\in\Lang(A_0,q_F)\) and the~\reftdta{R-closed} condition of \(A_0\).
Then \(u_0\in\Lang(A,q_F)\) holds 
due to \(u_0\in\Lang(A_0,q_F)\subseteq L_F\) and \refLF{posLF}.
Since \(q_F\) is a sink state in \(A\), \(u\in\Lang(A,q_F)\) holds.
In the latter case,
\(u\in\Lang(A,q_F)\) holds 
because 
\(t_0\in\Lang(A_0,q_F)\subseteq L_F\) implies \(t_0\in\Lang(A,q_F)\) by \refLF{posLF}
and
\(q_F\) is a sink state in \(A\).
Therefore, the~\reftdtas{R-closed} condition holds for \(A\) in the case of \(q=q_F\).
In the case of \(q\in Q_0\setminus Q_F\),
we have \(t\in\Lang(A_0,q)\) by \refLF{midLF},
hence \(u\in\Lang(A_0,q)\) holds by the~\reftdta{R-closed} condition of \(A_0\).
Since we have \(u\in\Lang(A,q)\cup\Lang(A,q_F)\) by \refLF{negLF},
the~\reftdtas{R-closed} condition holds for \(A\)
in the case  of \(q\in Q_0\setminus Q_F\).
\end{proof}

Since we try to find a \TDAS in the ascending order of the number of states
as Endrullis and Zantema have done for a \TDA,
one of the \TDASs for a given TRS with the smallest number of states
can be found if it exists.
It might be possible to find a \TDAS even with a smaller number of states
than a \TDA because of the relaxed closure property~\reftdtas{R-closed}.
Our experiment results in~\cref{sec:results} do not show such a case, though.

\subsection{SAT encoding of termination-disproving tree automata}
Endrullis and Zantema showed that
the problem of finding \TDAs can be reduced to
the boolean satisfiability problem (SAT, for short)
by fixing the number of states of tree automata.
Although we essentially follow their method,
we will present a method which can find a \TDAS
more efficiently 
because of the restriction of its form.
%
The efficiency of the disproof search is improved 
not only by finding a \TDAS instead of a \TDA
but also by specializing their method 
for non-erasing sole combinatory calculus.


We present our SAT encoding method to be self-contained,
where the differences from the method by Endrullis and Zantema
(\emph{EZ method}, for short)
are explicitly explained for each step.
We first explain problem settings and definitions used in our method,
introduce propositional variables and their meaning,
and then show propositional formulas over them that must hold.

\subsubsection{Problem setting and definitions}
Let \(Z\) be a non-erasing combinator
whose termination is to be disproved.
Recall that the reduction rule of \(Z\) is
represented by a singleton TRS \(R_Z\) 
over \(\Sigma_Z=\{Z^{(0)}, \A^{(2)}\}\).
Let \(l_Z \to r_Z\) be the unique rule of \(R_Z\)
where 
\( l_Z = \A(\dots\A(\A(Z, x_1), x_2)\dots, x_{M_Z})\)
with some \(M_Z\geq1\) and
\(r_Z\) is built from the binary function symbol \(\A\) 
and variables \(x_1,x_2,\dots,x_{M_Z}\)
so that \(\FVar(l_Z)=\FVar(r_Z)\) holds for \(R_Z\) to be non-erasing.
We write \(U_Z\) for the set \(\tsub{l_Z}\cup\tsub{r_Z}\).
The \emph{left depth} \(\ldepth(t)\) of a \(Z\)-term \(t\in\Term{\Sigma_Z}\)
is defined by
\(\ldepth(Z) = 0\) and
\(\ldepth(\A(t_1,t_2)) = 1 + \ldepth(t_1)\).
The left depth of a term is useful
in determining whether the term is redex or not.
A \(Z\)-term \(t\) has the left depth \(M_Z\) if and only if
\(t\) is a redex of \(R_Z\).

Let \(A=\tuple{Q,\Sigma,\{q_F\},\Delta}\) 
be a \TDAS with a final sink state \(q_F\in Q\),
which is to be found if it exists.
We fix the number of states as \(|Q|=N\)
in our encoding.
We iteratively ask the SAT solver to find a \TDAS
increasing \(N\) one by one,
starting with \(N=M_Z+1\)
because no automaton with less than \(M_Z+1\) states can recognize
the existence of a redex of \(R_Z\).
We use a function \(\minld_A: Q\to\Nat\) 
defined by
\(\minld_A(q)=\min_{t\in\Lang(A,q)}\ldepth(t)\).
The function is total
because the \TDAS to be found has only reachable states.


\subsubsection{Propositional variables}
Our SAT encoding involves 
three classes of propositional variables.
The first one has the form of
either \(v_{\Aa(q_1,q_2)\tto q}\) or \(v_{Z\tto q}\)
with \(q_1, q_2, q\in Q\),
which identify \(\Delta\).
These variables are expected to satisfy
\begin{quote}
\(v_{\Aa(q_1,q_2)\tto q}\) is true iff \(\A(q_1,q_2)\tto q\in\Delta\),
for all \(q_1, q_2, q\in Q\), and
\\[1ex]
\(v_{Z\tto q}\) is true iff \(Z\tto q\in\Delta\),
for all \(q\in Q\).
\end{quote}
This class of variables has been employed in the EZ method.

The second class of propositional variables has the form of
either \(v_{q,m}\) 
with \(q\in \QF\) and \(m<M_Z\)
or \(v_{q,\redex}\) 
with \(q\in Q\).
They are expected to satisfy
\begin{quote}
\(v_{q, m}\) is true iff 
\(\displaystyle m = \minld_A(q)\),
for all \(q\in \QF\) and \(m < M_Z\), and
\\[1ex]
\(v_{q,\redex}\) is true only if \(\Lang(A,q)\cap\NF(R_Z)=\emptyset\),
for all \(q\in Q\).
\end{quote}
This class of variables is newly introduced in our method
in order to check the existence of redexes.
Instead,
the EZ method employs 
propositional variables that represent
reachability of states of a tree automaton
obtained by product construction of two tree automata, 
\(A\) and \(B\),
where \(B\) accepts all terms in normal form of \(R_Z\).

The third class of propositional variables has the form of
\(v_{t,\alpha,q}\) 
with a term \(t\in U_Z\), 
a substitution \(\alpha:\Var(t)\to \QF\), and a state \(q\in Q\).
Note that the number of possible substitutions \(\alpha\) is finite
because \(U_Z\) and \(Q\) are finite.
Therefore the number of this class of variables is also finite.
These variables are expected to satisfy
\begin{quote}
\(v_{t,\alpha,q}\) is true iff \(t\alpha\tto^*q\)
\end{quote}
for all \(t\in U_Z\), \(\alpha:\Var(t)\to \QF\), and \(q\in Q\).
This class of variables has been employed in the EZ method,
with the difference of the domain of substitutions.
In their method,
each substitution is defined over the set of all states
including final states,
while our encoding excludes the final state from the domain of substitutions.
The difference makes the number of this class of variables much smaller,
which will reduce the number of clauses 
passed to the SAT solver.
We will explain later
the reason why the final state can be left out in our encoding method.

Besides the three classes of propositional variables,
our implementation of the proposed encoding method employs
extra variables which is equivalent to conjunction of the other variables.
They are introduced in order for the number of clauses to be smaller
using a standard technique 
of Tseitin transformation~\cite{Tseitin1983},
where sub-formulas are replaced by new propositional variables
to avoid exponential brow-up of the number of clauses%
.
The original method by Endrullis and Zantema may also have used the technique
thought it is not explicitly mentioned in their article.
Their implementation is currently no longer available,
so we cannot be sure how they actually do it.

\subsubsection{Propositional formulas}
Recall that 
a \TDAS \(A=\tuple{Q,\Sigma_Z,\{q_F\},\Delta}\)
with \(|Q|=N\) is to be found by the SAT solver.
The Boolean values of the propositional variables introduced so far
exactly identify \(\Delta\).
It must also be ensured, however, 
that there is no inconsistency in the valuation of propositional variables,
and that \(A\) satisfies the \TDAS conditions.
The propositional formulas to be satisfied will be shown in sequence
condition by condition.
It is assumed that they will eventually be combined together by conjunction
and passed to the SAT solver to find an appropriate valuation.

Firstly, 
the final state \(q_F\) must be sink according to the~\reftdtas{non-empty} condition.
Hence the following propositional formulas must hold:
\begin{equation}
\bigwedge_{q_F\in\{q_1,q_2\}\subseteq Q}
\bigg(
  v_{\Aa(q_1,q_2)\tto q_F}
  \land
  \bigwedge_{q\in\QF}
  \lnot v_{\Aa(q_1,q_2)\tto q}
\bigg)
\end{equation}
All of these propositional variables 
are immediately forced to be assigned to either true or false
in the phase of unit propagation of the SAT solver.
In addition, the~\reftdtas{non-empty} condition requires that
\(q_F\) is reachable.
Since we will find a \TDAS with the smallest number of states,
we assume that all states in \(Q\) are reachable.
The reachability of states can be paraphrased as
the existence of a total order on \(Q\) with appropriate properties
as considered in the EZ method.
We will employ the following lemma to specify the total order.
%
The proof is similar to that of the EZ method
except that the final state is fixed as a sink state in our setting.
\begin{onlyproc}%
The details of the proof are found
in the full version of this paper~\cite{DBLP:journals/corr/Nakano24ciaafull}.
\end{onlyproc}

\begin{lemm}\labelx{lem:non-empty}
Let \(A=\tuple{Q,\Sigma_Z,\{q_F\},\Delta}\) be a tree automaton
with a sink state \(q_F\).
Then, all states in \(Q\) are reachable if and only if
there exists a total order \(<\) on \(Q\)
with maximal element \(q_F\)
such that 
for every \(q\in Q\) 
there exists either
\(Z\tto q\in\Delta\) 
or
\(\A(q_1,q_2)\tto q\in\Delta\) with \(q_1<q\) and \(q_2<q\).
\end{lemm}
%
\begin{onlyfull}
\begin{proof}
Firstly, the `if'-statement is shown.
Suppose that there exists an unreachable state.
Let \(q\in Q\) be the smallest unreachable state in the order \(<\).
From the assumption, there exists either \(Z\tto q\in\Delta\) or
\(\A(q_1,q_2)\tto q\in\Delta\) with \(q_1<q\) and \(q_2<q\).
In the former case, \(q\) is reachable, which is a contradiction.
In the latter case, 
if both \(q_1\) and \(q_2\) are reachable,
then \(q\) is reachable, which is a contradiction;
if either \(q_1\) or \(q_2\) is unreachable,
then it contradicts the minimality of \(q\).
Therefore, all states in \(Q\) are reachable.

Secondly, the `only if'-statement is shown.
Suppose that all states in \(Q\) are reachable.
When \(Q\) is a singleton set of \(q_F\), the statement holds trivially,
hence we assume that \(Q\) contains at least two states.
Let \(Q_0=\{q\in Q\setminus\{q_F\} \mid Z\tto q\in\Delta\}\).
Then \(Q_0\) is not empty, 
otherwise no state in \(Q\setminus\{q_F\}\) is reachable
due to the sink condition of \(q_F\).
We can arbitrarily define an order \(<\) on \(Q_0\),
and thus let \(\{q_1,\dots,q_{|Q_0|}\}\) be the set \(Q_0\)
with \(q_i<q_{i+1}\) for \(1\le i<|Q_0|\).
Starting from \(Q_0\), 
repeatedly adding a state \(q (\ne q_F)\) not in the set 
but of \(\A(q_1,q_2)\tto q\in\Delta\) with \(q_1\) and \(q_2\) in the set
eventually yields \(Q\setminus\{q_F\}\)
due to the reachability of all states.
The iteration gives a sequence of sets of states as
\(Q_0=Q_{|Q_0|}\subsetneq Q_{|Q_0|+1}\subsetneq\dots
\subsetneq Q_{|Q|-1}=Q\setminus\{q_F\}\),
where adjacent sets differ only by one element.
Let \(\{q_i\}\) be \(Q_i\setminus Q_{i-1}\) for every \(|Q_0|<i<|Q|\)
and let \(q_i<q_{i+1}\) be defined for all \(|Q_0|\le i<|Q|-1\).
In addition, we define \(q_{|Q|-1}<q_F\)
because of the reachability of \(q_F\).
Thereby, the order \(<\) obviously satisfies the condition
to conclude the `only if'-statement.
\end{proof}
\end{onlyfull}

\Cref{lem:non-empty} makes it easy to find a \TDAS
in which all states are reachable.
Without loss of generality,
we can fix an ordered sequence of all states in \(Q\)
as \(p_1<p_2<\dots<p_N=q_F\) (recall \(|Q|=N\))
and force the states to satisfy the property in \Cref{lem:non-empty}.
The restriction is directly encoded by
\begin{equation}
\bigwedge_{q\in Q}\;
\bigg(
v_{Z\tto q}
\lor
\bigvee_{q_1, q_2<q}\;
v_{\Aa(q_1,q_2)\tto q}
\bigg)
\text,
\end{equation}
which has been employed
also in the EZ method,
although our encoding differs in that 
the ordered sequence of states should end with the unique final state.


Secondly, \(A\) must satisfy the~\reftdtas{no-nf} condition,
which requires that \(A\) accepts no terms in normal form.
Since \(\Lang(A)=\Lang(A,q_F)\),
a propositional formula
\begin{equation}
v_{q_F,\redex}
\end{equation}
(which is just a unit clause) should hold.
To ensure that the Boolean value of this variable is valid,
all of the propositional variables of the form
\(v_{q,m}\) 
or \(v_{q,\redex}\) 
should be properly assigned as follows.
%
The propositional variables \(v_{q, m}\) for \(q\in\QF\) and \(m<M_Z\)
are expected to satisfy \(m = \minld_A(q)\).
It is easy to see that \(\minld_A(q)\) can be effectively computed 
for each \(q\in Q\) from the transition rules in \(\Delta\).
The next lemma claims this fact as a logical statement
so as to be used for giving appropriate propositional formulas.
\begin{onlyproc}
The proof is found in the full paper~\cite{DBLP:journals/corr/Nakano24ciaafull}.
\end{onlyproc}
%
\begin{lemm}\labelx{lem:ldepth}
Let \(A=\tuple{Q,\Sigma_Z,\{q_F\},\Delta}\) be a tree automaton
with a combinator \(Z\)
where all states in \(Q\) are reachable.
Then for every \(q\in Q\) and \(m\in\Nat\),
\begin{inparaenum}[~(\ref{lem:ldepth}-1)]
\item\label{item:ld-zero}%
\(\minld_A(q)=0\) if and only if \(Z\tto q\in\Delta\), and
\item\label{item:ld-succ}%
\(\minld_A(q)=m>0\) if and only if
there is neither \(Z\tto q\in\Delta\) nor \(\A(q_1,q_2)\tto q\in\Delta\)
with \(\minld_A(q_1)<m-1\)
and
there exists \(\A(q_1,q_2)\tto q\in\Delta\) with \(\minld_A(q_1)=m-1\).
\end{inparaenum}
\end{lemm}
\newcommand\refminld[1]{(\ref{lem:ldepth}-\ref{item:ld-#1})}
%
\begin{onlyfull}
\begin{proof}
The~\refminld{zero} statement is obvious by the definitions of \(\minld\) and \(\ldepth\).
For the~\refminld{succ} statement,
we split the latter part into three parts:
\begin{enumerate}[~(\ref{lem:ldepth}-\ref{item:ld-succ}-1)]
\item\label{mld-succ-Z}
  \(Z\tto q\not\in\Delta\),
\item\label{mld-succ-lt}
  \(\A(q_1,q_2)\tto q\not\in\Delta\) 
  for all \(q_1,q_2\in Q\) with \(\minld_A(q_1)<m-1\), and
\item\label{mld-succ-eq}
  \(\A(q_1,q_2)\tto q\in\Delta\) 
  for some \(q_1,q_2\in Q\) with \(\minld_A(q_1)=m-1\).
\end{enumerate}
\newcommand\refmldsucc[1]{(\ref{lem:ldepth}-\ref{item:ld-succ}-\ref{mld-succ-#1})}
Firstly, we show the `if'-statement of \refminld{succ}.
Let \(q\in Q\) and \(m\in\Nat\) be fixed
and suppose that \refmldsucc{Z}, \refmldsucc{lt}, and \refmldsucc{eq} hold.
By \refmldsucc{eq},
there exists \(t_1\in\Lang(A,q_1)\) such that 
\(\A(q_1,q_2)\tto q\in\Delta\) and \(\ldepth(t_1)=m-1\).
By taking \(t_2\in\Lang(A,q_2)\) because of the reachability of \(q_2\),
we have \(\A(t_1,t_2)\in\Lang(A,q)\) whose left depth is \(m\),
hence \(\minld_A(q)\le m\).
Assume that \(\minld_A(q)<m\).
Then there exists \(t\in\Lang(A,q)\) such that \(\ldepth(t)<m\).
Since we have \(t\ne Z\) by \refmldsucc{Z},
\(t=\A(t_1,t_2)\) holds with some \(t_1,t_2\in\Term{\Sigma_Z}\).
Then, there exists \(q_1,q_2\in Q\) such that
\(\A(q_1,q_2)\tto q\in\Delta\), \(t_1\in\Lang(A,q_1)\) and \(t_2\in\Lang(A,q_2)\).
From \(\ldepth(t)<m\), we have \(\ldepth(t_1)<m-1\),
which contradicts \refmldsucc{lt}.
Therefore, \(\minld_A(q)=m\) holds.
Due to \refmldsucc{Z}, we have \(\minld_A(q)>0\),
hence the `if'-statement of \refminld{succ} holds.
%

Next, we show the `only if'-statement of \refminld{succ}.
Suppose that \(\minld_A(q)=m>0\) holds with \(q\in Q\) and \(m\in\Nat\).
The~\refmldsucc{Z} statement is immediate from \(\minld_A(q)>0\).
We prove the~\refmldsucc{lt} statement by contradiction.
Suppose that there exists \(\A(q_1,q_2)\tto q\in\Delta\) 
such that \(\minld_A(q_1)<m-1\).
Then we have \(t_1\in\Lang(A,q_1)\) with \(\ldepth(t_1)<m-1\).
By taking \(t_2\in\Lang(A,q_2)\) because of the reachability of \(q_2\),
we have \(t=\A(t_1,t_2)\in\Lang(A,q)\) with \(\ldepth(t)<m\),
which contradicts the assumption of \(\minld_A(q)=m\).
The~\refmldsucc{eq} statement is shown as follows.
From the assumption of \(\minld_A(q)=m\),
there must be \(t\in\Lang(A,q)\) such that \(\ldepth(t)=m\).
Since \(m>0\),
we have \(t=\A(t_1,t_2)\) with some \(t_1,t_2\in\Term{\Sigma_Z}\).
Then there exists \(\A(q_1,q_2)\tto q\in\Delta\) such that
\(t_1\in\Lang(A,q_1)\) and \(t_2\in\Lang(A,q_2)\).
Because of \(\ldepth(t)=m\), we have \(\ldepth(t_1)=m-1\),
hence \(\minld_A(q_1)\le m-1\).
Therefore, we have \(\minld_A(q_1)=m-1\) from \refmldsucc{lt},
and thus \refmldsucc{eq} holds,
which concludes the `only if'-statement of \refminld{succ}.
\end{proof}
\end{onlyfull}

We only need the lemma for non-final states 
though it holds even for the final state \(q_F\)
because the propositional variables \(v_{q,m}\) are given only for \(q\in\QF\).
The statement \refminld{zero} indicates that
\(v_{q,0}\) for each \(q\in\QF\) has
an appropriate Boolean value 
by the following propositional formula:
\begin{equation}
\bigwedge_{q\in\QF}
\left(v_{q,0} \Iff v_{Z\tto q}\right)\text.
\end{equation}
We could use the same propositional variable for \(v_{q,0}\) and \(v_{Z\tto q}\)
in an efficient implementation, though.
Additionally,
the statement \refminld{succ} indicates that
the propositional variable \(v_{q,m}\) with \(m>0\) has 
an appropriate Boolean value 
by the following propositional formula:
\begin{equation}
\bigwedge_{q\in\QF}\,
\bigwedge_{m < M_Z}\,
\left(v_{q,m}\Iff
\left(
\lnot v_{Z\tto q}
\land P_1(q,m)
\land P_2(q,m)
\right)
\right)
\end{equation}
where 
\begin{align*}
P_1(q,m) & = 
\bigwedge_{0\le k< m-1}
\bigwedge_{q_1,q_2\in\QF}
\left(v_{\Aa(q_1,q_2)\tto q}\To \lnot v_{q_1,k}\right)\text{ and}
\\
P_2(q,m) & = 
\bigvee_{q_1,q_2\in\QF}
\left(v_{\Aa(q_1,q_2)\tto q}\land v_{q_1,m-1}\right)\text.
\end{align*}

The propositional variable \(v_{q,\redex}\) for each \(q\in Q\)
is expected to be true 
only if \(\Lang(A,q)\cap\NF(R_Z)=\emptyset\).
The next lemma is used for giving propositional formulas
in order for the variables \(v_{q,\redex}\) 
to have an appropriate Boolean value.
%
%
%

\begin{lemm}\labelx{lem:redex}
Let \(A=\tuple{Q,\Sigma_Z,\{q_F\},\Delta}\) be a tree automaton
with a sink state \(q_F\),
let \(R_Z\) be a TRS with a combinator \(Z\),
and let \(\Qrdx\subseteq Q\) 
be a set of states such that for all \(q\in\Qrdx\),
\begin{inparaenum}[~(\ref{lem:redex}-1)]
  \item\label{item:rdx-leaf}%
  \(Z\tto q\not\in\Delta\) holds, and
  \item\label{item:rdx-node}%
  \(\minld_A(q_1) = M_Z-1\) holds
  if there is \(\A(q_1,q_2)\tto q\in\Delta\) with \(q_1,q_2\in Q\setminus\Qrdx\).
\end{inparaenum}
Then, \(\Lang(A,q)\cap\NF(R_Z)=\emptyset\) holds for every \(q\in\Qrdx\).
\end{lemm}
\newcommand\refrdx[1]{(\ref{lem:redex}-\ref{item:rdx-#1})}

\begin{onlyfull}
\begin{proof}
  The conclusion of the statement can be rephrased as follows:
  for every \(t\in\Lang(A,q)\) with \(q\in\Qrdx\), 
  we have \(t\not\in\NF(R_Z)\).
  The statement is shown by induction on \(t\in\Term{\Sigma_Z}\).
  In the case of \(t=Z\), it immediately holds because of \refrdx{leaf},
  which implies there are no \(q\in\Qrdx\) such that \(Z\in \Lang(A,q)\).
  In the case of \(t=\A(t_1,t_2)\in\Lang(A,q)\) with some \(q\in\Qrdx\),
  there is \(\A(q_1,q_2)\tto q\in\Delta\) 
  such that \(t_1\in\Lang(A,q_1)\) and \(t_2\in\Lang(A,q_2)\) hold.
  If we have either \(q_1\in\Qrdx\) or \(q_2\in\Qrdx\),
  then \(t_1\not\in\NF(R_Z)\) or \(t_2\not\in\NF(R_Z)\) holds from the induction hypothesis.
  It indicates \(t\not\in\NF(R_Z)\) because \(t\) contains a redex of \(R_Z\).
  If we have \(q_1,q_2\in Q\setminus\Qrdx\),
  \(\minld_A(q_1)=M_Z-1\) holds by \refrdx{node}.
  It indicates \(\ldepth(t_1)\ge M_Z-1\), hence \(\ldepth(t)\ge M_Z\).
  Thus, \(t\not\in\NF(R_Z)\) holds.
\end{proof}
\end{onlyfull}
\begin{onlyproc}%
The proof is found in the full paper~\cite{DBLP:journals/corr/Nakano24ciaafull}.
\end{onlyproc}
Let \(\Qrdx\subseteq Q\) be the set of states \(q\)
such that \(v_{q,\redex}\) is true.
\Cref{lem:redex} guarantees that
\(v_{q,\redex}\) is true only if \(\Lang(A,q)\cap\NF(R_Z)=\emptyset\)
whenever \(\Qrdx\) satisfies the conditions \refrdx{leaf} and \refrdx{node},
that is, the following propositional formulas should be true:
for \refrdx{leaf},
\begin{equation}
  \bigwedge_{q\in Q}\left(
v_{q,\redex} \To \lnot v_{Z\tto q}
  \right)\text;
\end{equation}
for \refrdx{node}, 
\begin{equation}
\bigwedge_{q\in Q}\,
\bigwedge_{q_1,q_2\in \QF}\left(
v_{q,\redex}\land v_{\Aa(q_1,q_2)\tto q}
\land\lnot v_{q_1,\redex}\land\lnot v_{q_2,\redex}\To v_{q_1,M_Z-1}
\right)\text.
\end{equation}

Finally, we need to ensure that \(A\) is almost closed under \(R_Z\)
in the sense of the~\reftdta{R-closed} condition.
As the EZ method does,
we use the following lemma which gives a procedure
to have the condition.
While the EZ method employs the existing results%
~\cite[Proposition 12]{DBLP:conf/rta/Genet98},
we give a proof of this lemma
\begin{onlyproc}%
in the full paper~\cite{DBLP:journals/corr/Nakano24ciaafull}
\end{onlyproc}
because the~\reftdtas{R-closed} condition is different from the closure property.
\begin{lemm}\labelx{lem:closed}
  Let \(A=\tuple{Q, \Sigma, \{q_F\}, \Delta}\) be a tree automaton
  with a sink state \(q_F\),
  and let \(R\) be a left-linear TRS over \(\Sigma\).
  Then, \(s\in\Lang(A,q)\) implies \(t\in\Lang(A,q)\cup\Lang(A)\)
  for all \(q\in Q\) and \(s,t\in\Term{\Sigma}\) with \(s\toi_R t\)
  if \(l\alpha\tto_\Delta^* q\) implies \(r\alpha\tto_\Delta^* q\) or 
  \(r\alpha\tto_\Delta^* q_F\)
  for all \(q\in Q\), \(l\to r\in R\) and \(\alpha:\FVar(l)\to Q\).
\end{lemm}
\begin{onlyfull}
\begin{proof}
Let \(s\in\Lang(A,q)\) with \(s\toi_R t\).
By definition, there exists \(C\in\Ctx{\Sigma}\) 
such that \(s=C[l\sigma]\) and \(t=C[r\sigma]\)
with \(l\to r\in R\) and \(\sigma:\FVar(l)\to\Term{\Sigma}\).
We show the statement by induction on the structure of \(C\).
In the case of \(C=\hole\), 
we have \(s = l\sigma\tto^* q\) and \(t = r\sigma\).
Let \(\alpha: \FVar(l)\to Q\) be a function defined
so that \(\sigma(x)\tto^*\alpha(x)\) 
is used in the derivation of \(l\sigma\tto^* q\).
Note that \(\alpha\) is well-defined because of the left-linearity of \(R\).
Then \(l\alpha\tto^* q\) holds,
hence \(r\alpha\tto^* q\) or \(r\alpha\tto^* q_F\) from the assumption.
Therefore, \(r\sigma\tto^* q\) or \(r\sigma\tto^* q_F\) holds,
which indicates \(t\in\Lang(A,q)\cup\Lang(A)\).
In the case of \(C=f(t_1,\dots,C',\dots,t_n)\)
where \(C'\) is the \(j\)-th child of \(f\in\Sigma^{(n)}\) with \(1\le j\le n\),
there exists \(q_1,\dots,q_n\in Q\)
such that \(f(q_1,\dots,q_n)\tto q\in\Delta\),
\(t_i\in\Lang(A,q_i)\) for \(1\le i\le n\) with \(i\ne j\), and
\(C'[l\sigma]\in\Lang(A,q_j)\).
From \(C'[l\sigma]\toi_R C'[r\sigma]\) and the induction hypothesis,
we have \(C'[r\sigma]\in\Lang(A,q_j)\cup\Lang(A)\).
When \(C'[r\sigma]\in\Lang(A,q_j)\),
we have \(t=f(t_1,\dots,C'[r\sigma],\dots,t_n)\in\Lang(A,q)\)
using the same transition rule.
When \(C'[r\sigma]\in\Lang(A)=\Lang(A,q_F)\),
we have \(t=f(t_1,\dots,C'[r\sigma],\dots,t_n)\in\Lang(A,q_F)=\Lang(A)\)
because \(q_F\) is sink.
Therefore, \(t\in\Lang(A,q)\cup\Lang(A)\) holds.
\end{proof}
\end{onlyfull}


Recall that the propositional variable \(v_{t,\alpha,q}\) is true
when a term \(t\) is accepted with \(q\in Q\)
if every variable \(x\) in \(t\) 
is substituted with a term accepted with a state \(\alpha(x)\).
Following the EZ method, 
the procedure of computing the state of each subterm 
is simulated by imposing relations among the propositional variables.
Unlike their method,
we only need to consider the case where the range of \(\alpha\) is \(\QF\)
since the closure property always holds 
if \(\alpha(x)=q_F\) for some \(x\) because \(q_F\) is sink and \(R_Z\) is non-erasing.
This arrangement would substantially reduce the number of propositional variables.
For each subterm \(t=\A(t_1,t_2)\in U_Z\), we should have
\begin{equation}
  \bigwedge_{q\in Q}\,
  \bigwedge_{\alpha:\FVar(t)\to \QF}\!\!
  \bigg(
  v_{t,\alpha,q} \Iff
  \bigvee_{q_1,q_2\in Q}\!\!\left(
  v_{t_1,\alpha[t_1],q_1} \land v_{t_2,\alpha[t_2],q_2} \land v_{\Aa(q_1,q_2)\tto q}
  \right)\!
  \bigg)
\end{equation}
where \(\alpha[t]\) stands for the substition \(\alpha\) 
whose domain is restricted to \(\FVar(t)\).
For each leaf in the subterms \(U_Z\), we should have
\begin{equation}
\bigwedge_{q\in Q}
\bigg(
\left(
v_{Z,\emptyset,q} \Iff v_{Z\tto q}
\right)
\land
\bigwedge_{x\in\FVar(l_Z)}\,
\Big(
  v_{x,\{x\mapsto q\},q}
  \land
  \bigwedge_{q'\in Q\setminus\{q\}}
  \lnot v_{x,\{x\mapsto q'\},q}
\Big)
\bigg)
\end{equation}
where the former part indicates that
we could use the same propositional variable for
\(v_{Z,\emptyset,q}\) and \(v_{Z\tto q}\).
And then,
the~\reftdtas{R-closed} condition requires 
the propositional formula
\begin{equation}
\bigwedge_{\alpha:\FVar(l_Z)\to Q}\,
\left(
v_{l_Z,\alpha,q} \To v_{r_Z,\alpha,q}\lor v_{r_Z,\alpha,q_F}
\right)\text.
\end{equation}

%
%
%
%

%
%
%
%
%
%
%
%
%
%

\subsection{Applying to non-termination of specific combinators}
\label{sec:results}
We have shown how to construct propositional logic formulas
for the existence of \TDAS 
with fixed number of states 
for a sole combinatory calculus.
We have implemented the construction 
and input the obtained formulas to a SAT solver 
to examine termination for combinators shown in~\cref{fig:combs},
for which termination is unknown.
We confirmed that our method can efficiently find a \TDAS
compared to the method developed by Endrullis and Zantema for a \TDA.
Note that their implementation is not currently publicly available,
so we have re-implemented it.
For fairness, we have also applied the same optimization 
(e.g., Tseitin transformation)
as in the implementation of our construction.
Our implementation is written in OCaml and uses the Kissat SAT solver.
We also implement the ability to output the smallest term accepted
by the found \TDAS,
thus allowing the output of a non-normalizable term.

The results of the examination for each combinator are 
shown in~\cref{tab:results}.
\begin{onlyfull}%
The details of the \TDAS obtained are given
in~\cref{sec:tdta}
\end{onlyfull}
\begin{onlyproc}%
Due to page limitation,
the details of the \TDAS obtained are found
in the full paper~\cite{DBLP:journals/corr/Nakano24ciaafull}
\end{onlyproc}
with a found counterexample for termination.
Here we only show the number of states of the \TDAS, 
the size of the propositional logic formula used
(number of propositional variables and clauses),
and the computation time (the sum of encoding and solving time)
including the attempts to fewer states.
\begin{table}[t]\centering
  \caption{Results of disproving termination of combinators given in~\cref{fig:combs}.}
  \label{tab:results}
  \newcommand\hdrs[1]{%
    & \mc1{c}{\#\(Q\)} & \mc1{c}{\#vars} 
    & \mc1{c}{\#clauses} & \mc1{c#1}{time (s)}%
  }
  \begin{tabular}{@{~~~}l|r@{~~~}r@{~~~}r@{~~~}r@{~~~}|r@{~~~}r@{~~~}r@{~~~}r@{~~~}}
    \hline
    \multirow{2}{*}{\qquad} &
    \mc4{c|}{EZ method} & \mc4c{Our method} 
    \\\cline{2-9} \hdrs{|} \hdrs{}
    \\\hline 
    \(\cP\) 
    & ~~~4 & 15,648 & 75,237 & 0.71
    & ~~~4 &  7,116 & 33,543 & 0.22
    \\
    \(\cP_3\) 
    & 6 & 161,544 & 793,351 & 11.5
    & 6 &  96,058 & 469,233 & 5.6
    \\
    \(\!\cD_1\) 
    & 6 & 936,810 & 4,594,129 & 97.0
    & 6 & 462,528 & 2,264,943 & 32.6
    \\
    \(\!\cD_2\)
    & 6 & 936,810 & 4,594,129 & 96.6
    & 6 & 462,528 & 2,264,943 & 32.9
    \\
    \(\!\!\:\cPh\) 
    & 7 & 1,975,302 & ~~9,741,173 & 260.6
    & 7 & 1,099,527 & 5,416,581 & 108.3
    \\
    \(\!\!\:\cPh_2\)
    & -- & ---~~ & ---~~ & ---~
    & 9 & 55,695,683 & 276,052,931 & 39,268.4
    \\ 
    \(\cS_1\) 
    & 7 & 1,975,302 & 9,741,173 & 262.1
    & 7 & 1,099,527 & 5,416,581 & 111.6
    \\
    \(\cS_2\) 
    & 6 & 745,002 & 3,655,825 & 81.6
    & 6 & 379,278 & 1,857,693 & 30.3
    \\\hline
  \end{tabular}
\end{table}
%
The `EZ method' column shows the results of (our re-implementation of)
the original method by Endrullis and Zantema;
the `Our method' column shows the results of our construction method
introduce in the present paper.
The EZ method failed to find a \TDA
for \(\cPh_2\) within 24-hour computation.
Note that the number of propositional variables and clauses
is the same for different combinators.
This is because they depend only on the number of states, 
the number of variables in the left-hand side of the combinator's rule
and the number of subterms in the right-hand side.
A major improvement in our method is 
a reduction of the number of variables of the form \(v_{t,\alpha,q}\)
by restricting the range of \(\alpha\),
which is the most significant part of the SAT encoding.
The number of possible substitutions is reduced 
from \(|Q|^{N}\) to \((|Q|-1)^N\),
where \(N\) is the number of variables in the left-hand side of the rule.
Our method can indeed generate less propositional logic formulas
than the original one,
and also succeed in disproving the termination of more combinators.
Both method failed to find the disproof of termination 
for the remaining combinators \(\cS_3\) and \(\cS_4\) in~\cref{fig:combs}.
We have confirmed that their termination cannot be disproven
by a \TDAS with at most 9 states.


%

\section{Concluding remark}
We have proposed a method to disprove the termination of 
sole combinatory calculi with tree automata
by extending the method proposed by Endrullis and Zantema.
Specifically,
we have shown that a tree automaton with a final sink state is sufficient
to disprove termination of non-erasing combinators.
We have succeeded in disproving the termination of 8 combinators
which is unknown of their termination found in Smullyan's book.
The remaining two combinators \(\cS_3\) and \(\cS_4\)
still have unknown termination.
We may need more improvement to disprove their termination.
However,
our method in which
termination is disproved by a tree automaton with a final sink state
may be applicable to other non-erasing term rewriting systems.
It would be interesting to investigate how effectual our method is
in more general settings.

\subsubsection*{Acknowledgments.}
The authors would like to thank the anonymous reviewers for their valuable comments.
This work was partially supported by JSPS KAKENHI Grant Numbers,
JP21K11744, JP22H00520, and JP22K11904.

%
%
\bibliographystyle{splncs04}
\bibliography{refs}

\begin{onlyfull}
\appendix
\newcommand\Dsink[1]{\Delta_{\textsf{sink},#1}}
\section{The obtained \TDASs for combinators}
\label{sec:tdta}
We show the details of \TDASs
for combinators obtained by our implemetation of the method
presented in Section~\ref{sec:disproof}.
The implication gives one of the smallest terms accepted by a \TDAS
as a counterexample to the termination of the combinator.
Here, we show the definition of the \TDAS and a counterexample
for each combinator.
For a sink state \(q_F\),
we write \(\Dsink{q_F}\) for
the set consisting of transition rules \(\A(q_1,q_2)\tto q_F\)
with states \(q_1\) and \(q_2\) one of which is \(q_F\).
We write \(q_1\lor\dots\lor q_n\)
the right-hand side of non-deterministic transition rules
with the same left-hand side.
The counterexample is shown in terms of combinatory calculus 
rather than TRS with \(\A\) for readability.

\subsection{Combinator \(\cP\)}
\label{sec:tdta-P}
Combinator \(\cP\) defined by \(\cP x y z \to z (x y z)\)
is shown to be non-terminating by a \TDAS
\(A_{\cP}=\tuple{\{1,\ab 2,\ab 3,\ab 4\},\Sigma_{\cP},\{4\},\Dsink{4}\cup\Delta}\) where
\(\Delta\) consists of
\begin{align*}
& \cP \tto 1, && \A(1, 1) \tto 1\lor2, && \A(1, 2) \tto 1\lor2,\\
& \A(1, 3) \tto 1\lor2, && \A(2, 1) \tto 1\lor2, && \A(2, 2) \tto 1\lor2\lor3,\\
& \A(2, 3) \tto 1\lor2\lor3, && \A(3, 1) \tto 1\lor2, && \A(3, 2) \tto 1\lor2\lor3, \text{and}\\
& \A(3, 3) \tto 1\lor2\lor3\lor4\text.
\end{align*}
A counterexample of termination is
\(\cP\ab \cP\ab (\cP\ab \cP\ab) (\cP\ab \cP\ab (\cP\ab \cP\ab))\).

\subsection{Combinator \(\cP_3\)}
\label{sec:tdta-P3}
Combinator \(\cP_3\) defined by \(\cP_3 x y z \to y (x z y)\)
is shown to be non-terminating by a \TDAS
\(A_{\cP_3}=\tuple{\{1,\ab 2,\ab 3,\ab 4,\ab 5,\ab 6\},\Sigma_{\cP_3},\{6\},\Dsink{6}\cup\Delta}\) where
\(\Delta\) consists of
\begin{align*}
& \cP_3 \tto 1, && \A(1, 1) \tto 2, && \A(1, 2) \tto 2,\\
& \A(1, 3) \tto 2, && \A(1, 4) \tto 2\lor4, && \A(1, 5) \tto 2\lor3\lor4,\\
& \A(2, 2) \tto 2\lor3, && \A(2, 3) \tto 3\lor4, && \A(2, 4) \tto 2\lor3\lor4,\\
& \A(2, 5) \tto 5, && \A(3, 2) \tto 2\lor3, && \A(3, 3) \tto 3,\\
& \A(3, 4) \tto 2\lor3\lor4\lor5, && \A(3, 5) \tto 2\lor4, && \A(4, 2) \tto 2\lor5,\\
& \A(4, 3) \tto 3, && \A(4, 4) \tto 4\lor5, && \A(4, 5) \tto 5,\\
& \A(5, 2) \tto 6, && \A(5, 3) \tto 2\lor4, && \A(5, 4) \tto 5\lor6, \text{and}\\
& \A(5, 5) \tto 2\lor3\lor5\lor6\text.
\end{align*}
A counterexample of termination is
\(\cP_3\ab \cP_3\ab (\cP_3\ab \cP_3\ab (\cP_3\ab \cP_3\ab)) (\cP_3\ab \cP_3\ab) (\cP_3\ab \cP_3\ab)\).

\subsection{Combinator \(\cD_1\)}
\label{sec:tdta-D1}
Combinator \(\cD_1\) defined by \(\cD_1 x y z w \to x z (y w) (x z)\)
is shown to be non-terminating by a \TDAS
\(A_{\cD_1}=\tuple{\{1,\ab 2,\ab 3,\ab 4,\ab 5,\ab 6\},\Sigma_{\cD_1},\{6\},\Dsink{6}\cup\Delta}\) where
\(\Delta\) consists of
\begin{align*}
& \cD_1 \tto 1, && \A(1, 1) \tto 2, && \A(1, 2) \tto 2, && \A(1, 3) \tto 3,\\
& \A(1, 4) \tto 3, && \A(2, 2) \tto 3, && \A(2, 3) \tto 3, && \A(2, 4) \tto 3,\\
& \A(3, 3) \tto 3\lor4, && \A(3, 4) \tto 3\lor4, && \A(4, 3) \tto 4\lor5, && \A(4, 4) \tto 3, \text{and}\\
& \A(5, 4) \tto 6\text.
\end{align*}
A counterexample of termination is
\(\cD_1\ab \cD_1\ab (\cD_1\ab \cD_1\ab) (\cD_1\ab \cD_1\ab (\cD_1\ab \cD_1\ab)) (\cD_1\ab \cD_1\ab (\cD_1\ab \cD_1\ab))
      (\cD_1\ab \cD_1\ab (\cD_1\ab \cD_1\ab)
      (\cD_1\ab \cD_1\ab (\cD_1\ab \cD_1\ab)))\).

\subsection{Combinator \(\cD_2\)}
\label{sec:tdta-D2}
Combinator \(\cD_2\) defined by \(\cD_2 x y z w \to x w (y z) (x w)\)
is shown to be non-terminating by a \TDAS
\(A_{\cD_2}=\tuple{\{1,\ab 2,\ab 3,\ab 4,\ab 5,\ab 6\},\Sigma_{\cD_2},\{6\},\Dsink{6}\cup\Delta}\) where
\(\Delta\) consists of
\begin{align*}
& \cD_2 \tto 1, && \A(1, 1) \tto 2, && \A(1, 2) \tto 2, && \A(1, 3) \tto 3,\\
& \A(1, 5) \tto 3, && \A(2, 2) \tto 3, && \A(2, 3) \tto 3, && \A(2, 5) \tto 3,\\
& \A(3, 3) \tto 3\lor4, && \A(3, 5) \tto 3\lor4, && \A(4, 3) \tto 5, && \A(4, 5) \tto 3,\\
& \A(5, 3) \tto 1\lor6, \text{and} && \A(5, 5) \tto 1\lor5\lor6\text.
\end{align*}
A counterexample of termination is
\(\cD_2\ab \cD_2\ab (\cD_2\ab \cD_2\ab) (\cD_2\ab \cD_2\ab (\cD_2\ab \cD_2\ab)) (\cD_2\ab \cD_2\ab (\cD_2\ab \cD_2\ab)) (\cD_2\ab \cD_2\ab (\cD_2\ab \cD_2\ab))\).

\subsection{Combinator \(\cPh\)}
\label{sec:tdta-Phi}
Combinator \(\cPh\) defined by \(\cPh x y z w \to x (y w) (z w)\)
is shown to be non-terminating by a \TDAS
\(A_{\cPh}=\tuple{\{1,\ab 2,\ab 3,\ab 4,\ab 5,\ab 6,\ab 7\},\Sigma_{\cPh},\{7\},\Dsink{7}\cup\Delta}\) where
\(\Delta\) consists of
\begin{align*}
& \cPh \tto 1, && \A(1, 1) \tto 2, && \A(1, 3) \tto 4, && \A(1, 4) \tto 2,\\
& \A(1, 6) \tto 4, && \A(2, 2) \tto 3, && \A(2, 6) \tto 3\lor5, && \A(3, 6) \tto 6,\\
& \A(4, 3) \tto 5, && \A(5, 3) \tto 6, && \A(5, 5) \tto 3, && \A(5, 6) \tto 6, \text{and}\\
& \A(6, 6) \tto 7\text.
\end{align*}
A counterexample of termination is
\(\cPh\ab (\cPh\ab \cPh\ab (\cPh\ab \cPh\ab)) (\cPh\ab \cPh\ab (\cPh\ab \cPh\ab)) (\cPh\ab \cPh\ab (\cPh\ab \cPh\ab))
      (\cPh\ab (\cPh\ab \cPh\ab (\cPh\ab \cPh\ab)) (\cPh\ab \cPh\ab (\cPh\ab \cPh\ab))
      (\cPh\ab \cPh\ab (\cPh\ab \cPh\ab)))\).

\subsection{Combinator \(\cPh_2\)}
\label{sec:tdta-Phi2}
Combinator \(\cPh_2\) defined by
\(\cPh_2  x y z  w_1 w_2 \to x  (y w_1 w_2)  (z w_1 w_2)\)
is shown to be non-terminating by a \TDAS
\(A_{\cPh_2}=\tuple{\{1,\ab 2,\ab 3,\ab 4,\ab 5,\ab 6,\ab 7,\ab 8,\ab 9\},\ab\Sigma_{\cPh_2},\ab\{9\},\ab \Dsink{9}\cup\Delta}\) where
\(\Delta\) consists of
\begin{align*}
& \cPh_2 \tto 1, && \A(1, 1) \tto 2, && \A(1, 3) \tto 4, && \A(2, 2) \tto 3,\\
& \A(2, 6) \tto 3, && \A(2, 7) \tto 3\lor6, && \A(3, 7) \tto 3\lor7, && \A(4, 4) \tto 5,\\
& \A(4, 6) \tto 3, && \A(4, 7) \tto 3, && \A(5, 2) \tto 6, && \A(6, 6) \tto 7,\\
& \A(6, 7) \tto 3, && \A(7, 7) \tto 8, \text{and} && \A(8, 7) \tto 4\lor9\text.
\end{align*}
A counterexample of termination is
\(\cPh_2\ab (\cPh_2\ab \cPh_2\ab (\cPh_2\ab \cPh_2\ab)) (\cPh_2\ab (\cPh_2\ab \cPh_2\ab (\cPh_2\ab \cPh_2\ab))) (\cPh_2\ab \cPh_2\ab)
      (\cPh_2\ab (\cPh_2\ab \cPh_2\ab (\cPh_2\ab \cPh_2\ab)) (\cPh_2\ab (\cPh_2\ab \cPh_2\ab (\cPh_2\ab \cPh_2\ab)))
      (\cPh_2\ab \cPh_2\ab))
      (\cPh_2\ab (\cPh_2\ab \cPh_2\ab (\cPh_2\ab \cPh_2\ab)) (\cPh_2\ab (\cPh_2\ab \cPh_2\ab (\cPh_2\ab \cPh_2\ab)))
       (\cPh_2\ab
       \cPh_2\ab)
      (\cPh_2\ab (\cPh_2\ab \cPh_2\ab (\cPh_2\ab \cPh_2\ab)) (\cPh_2\ab (\cPh_2\ab \cPh_2\ab (\cPh_2\ab \cPh_2\ab)))
      (\cPh_2\ab \cPh_2\ab)))
      (\cPh_2\ab (\cPh_2\ab \cPh_2\ab (\cPh_2\ab \cPh_2\ab)) (\cPh_2\ab (\cPh_2\ab \cPh_2\ab (\cPh_2\ab \cPh_2\ab)))
       (\cPh_2\ab
       \cPh_2\ab)
      (\cPh_2\ab (\cPh_2\ab \cPh_2\ab (\cPh_2\ab \cPh_2\ab)) (\cPh_2\ab (\cPh_2\ab \cPh_2\ab (\cPh_2\ab \cPh_2\ab)))
      (\cPh_2\ab \cPh_2\ab)))\).

\subsection{Combinator \(\cS_1\)}
\label{sec:tdta-S1}
Combinator \(\cS_1\) defined by \(\cS_1 x y z w \to x y w (z w)\)
is shown to be non-terminating by a \TDAS
\(A_{\cS_1}=\tuple{\{1,\ab 2,\ab 3,\ab 4,\ab 5,\ab 6,\ab 7\},\Sigma_{\cS_1},\{7\},\Dsink{7}\cup\Delta}\) where
\(\Delta\) consists of
\begin{align*}
& \cS_1 \tto 1, && \A(1, 1) \tto 2, && \A(1, 2) \tto 6,\\
& \A(1, 3) \tto 6, && \A(1, 4) \tto 6, && \A(1, 5) \tto 6,\\
& \A(1, 6) \tto 6, && \A(2, 2) \tto 3\lor6, && \A(2, 3) \tto 3\lor6,\\
& \A(2, 4) \tto 3\lor6, && \A(2, 5) \tto 3\lor6, && \A(2, 6) \tto 3\lor6,\\
& \A(3, 2) \tto 4, && \A(3, 3) \tto 2\lor3\lor6, && \A(3, 4) \tto 4,\\
& \A(3, 5) \tto 4\lor5, && \A(3, 6) \tto 2\lor6, && \A(4, 2) \tto 2\lor3\lor6,\\
& \A(4, 3) \tto 2\lor3\lor5, && \A(4, 4) \tto 2\lor4\lor5, && \A(4, 5) \tto 4\lor5,\\
& \A(4, 6) \tto 3\lor4\lor6, && \A(5, 2) \tto 2\lor4, && \A(5, 3) \tto 4,\\
& \A(5, 4) \tto 2\lor3\lor5, && \A(5, 5) \tto 7, && \A(5, 6) \tto 2\lor6,\\
& \A(6, 2) \tto 2\lor4, && \A(6, 3) \tto 2\lor4, && \A(6, 4) \tto 2\lor4,\\
& \A(6, 5) \tto 2\lor4, \text{and} && \A(6, 6) \tto 2\lor6\text.
\end{align*}
A counterexample of termination is
\(\cS_1\ab \cS_1\ab (\cS_1\ab \cS_1\ab) (\cS_1\ab \cS_1\ab) (\cS_1\ab \cS_1\ab (\cS_1\ab \cS_1\ab))
      (\cS_1\ab \cS_1\ab (\cS_1\ab \cS_1\ab) (\cS_1\ab \cS_1\ab)
      (\cS_1\ab \cS_1\ab (\cS_1\ab \cS_1\ab)))\).

\subsection{Combinator \(\cS_2\)}
\label{sec:tdta-S2}
Combinator \(\cS_2\) defined by \(\cS_2 x y z w \to x z w (y z w)\)
is shown to be non-terminating by a \TDAS
\(A_{\cS_2}=\tuple{\{1,\ab 2,\ab 3,\ab 4,\ab 5,\ab 6\},\Sigma_{\cS_2},\{6\},\Dsink{6}\cup\Delta}\) where
\(\Delta\) consists of
\begin{align*}
& \cS_2 \tto 1, && \A(1, 1) \tto 2, && \A(1, 2) \tto 1\lor2,\\
& \A(1, 3) \tto 2, && \A(1, 4) \tto 1\lor2, && \A(1, 5) \tto 2,\\
& \A(2, 1) \tto 2, && \A(2, 2) \tto 2\lor3, && \A(2, 3) \tto 2,\\
& \A(2, 4) \tto 1\lor4, && \A(2, 5) \tto 1\lor4, && \A(3, 1) \tto 2,\\
& \A(3, 2) \tto 1\lor2\lor4, && \A(3, 3) \tto 2, && \A(3, 4) \tto 1\lor2,\\
& \A(3, 5) \tto 1\lor4, && \A(4, 1) \tto 2\lor3, && \A(4, 2) \tto 1\lor2,\\
& \A(4, 3) \tto 2, && \A(4, 4) \tto 1\lor4\lor5, && \A(4, 5) \tto 1\lor5,\\
& \A(5, 1) \tto 1\lor3\lor4, && \A(5, 2) \tto 1\lor2\lor4, && \A(5, 3) \tto 6,\\
& \A(5, 4) \tto 1\lor4\lor5, \text{and} && \A(5, 5) \tto 1\lor4\lor5\text.
\end{align*}
A counterexample of termination is
\(\cS_2\ab \cS_2\ab (\cS_2\ab \cS_2\ab) (\cS_2\ab \cS_2\ab) (\cS_2\ab \cS_2\ab (\cS_2\ab \cS_2\ab) (\cS_2\ab \cS_2\ab)) (\cS_2\ab \cS_2\ab (\cS_2\ab \cS_2\ab))\).



\end{onlyfull}

\end{document}